\documentclass{llncs}
\usepackage{makeidx}  
\usepackage{graphicx}
\usepackage{psfrag}
\usepackage{amsfonts}
\usepackage[all]{xy}
\newcommand{\Ke}{\kern0.4em}
\newcommand{\ke}{\kern0.2em}

\newcommand{\disp}{\displaystyle}

\newcommand{\inv}{^{-1}}
\newcommand{\q}{\quad}

\newcommand{\bb}[1]{\mathbb#1}
\newcommand{\eu}[1]{{\mathfrak#1}}

\newcommand{\im}{\mathop{\rm image}\nolimits}

\newcommand{\rk}{\mathop{\rm rk}\nolimits}

\newcommand{\transp}{^{\rm t}}

\newcommand{\Mat}{\mathop{\rm Mat}\nolimits}

\newcommand{\Sym}{\mathop{\rm Sym}\nolimits}
\newcommand{\SL}{\mathop{\rm SL}\nolimits}
\newcommand{\GL}{\mathop{\rm GL}\nolimits}
\newcommand{\Sp}{\mathop{\rm Sp}\nolimits}
\newcommand{\grass}{\mathop{\eu{G}\eu{r}}\nolimits}
\newcommand{\lgrass}{\mathop{\eu{L}\eu{G}\eu{r}}\nolimits}

\newcommand{\by}{\!\times\!}

\renewcommand{\phi}{\varphi}

\newtheorem{thm}{Theorem}
\newtheorem{prop}[thm]{Proposition}

\newtheorem{cor}[thm]{Corollary}

\begin{document}
\frontmatter          
\pagestyle{headings}  
\addtocmark{Homology of Generic Stabilizer States} 
%
%
\mainmatter              
\title{Homology of Generic Stabilizer States}
\titlerunning{Generic Stabilizer States}  
%
\author{Klaus Wirthm\"uller}
\authorrunning{Klaus Wirthm\"uller}   
%
\tocauthor{Klaus Wirthm\"uller}
\institute{Fachbereich Mathematik der
           Technischen Universit\"at Kaiserslautern,\\
           Postfach 3049, 67653 Kaiserslautern, Germany\\
           \email{wirthm@mathematik.uni-kl.de}}

\maketitle              

\begin{abstract}
This work is concerned with multi-party stabilizer states in the sense of quantum information theory. We investigate the homological invariants for states of which each party holds a large equal number $N$ of quantum bits. We show that in many cases there is a generic expected value of the invariants\ts: for large $N$ it is approximated with arbitrarily high probability if a stabilizer state is chosen at random. The result suggests that typical entanglement of stabilizer states involves but the sets comprising just over one half of the parties.

Our main tool is the Bruhat decomposition from the theory of finite Chevalley groups.
\end{abstract}
\section{Introduction}
\label{int}
\noindent
Entangled multipartite states are central to quantum information theory. As a reasonably comprehensive understanding of such states has turned out to be largely illusory, interest has been focused on partial aspects, in particular on special classes of states that are easier to study. Among such classes that of stabilizer states has attracted attention due to the fact that it includes highly entangled states but, on the other hand, lends itself to a much simpler mathematical treatment than is possible for general quantum states.

The question addressed in this note has two roots. In \cite{wirthmueller} a series of invariants of stabilizer states has been introduced in an attempt to classify qualitative aspects of entanglement. The nature and construction of these invariants, which rely on ideas from algebraic geometry, suggest an uneven distribution of their values\ts: while a few special states are expected to have equally special invariants, the large majority is thought to have a common \textit{generic} value of them.

This view is supported by evidence from the first non-trivial invariant $H^2(|\psi\rangle)$. Its meaning is quite explicit since it just measures the number of all-party GHZ states that may be extracted from $|\psi\rangle$ \cite{bravyi3,wirthmueller}. Here the other root of our note comes in\ts:  the work \cite{smith2} includes a result according to which stabilizer states from which a significant number of GHZ states can be extracted become increasingly scarce as the number of qubits per party increases. This may be restated saying that the suitably normalised generic value of $H^2(|\psi\rangle)$ tends to zero in this situation (where each party holds the same large number of qubits).

Stabilizer states were first described in \cite{gottesman_96} in the context of error-correcting quantum codes, and the formalism has been repeatedly presented since \cite{calderbank4,gottesman_98,hostens3,knill,rains,wirthmueller}. The main point of interest here is that pure stabilizer states in an $l$-qubit Hilbert space $\cal H$ are parametrised by the Lagrangian subspaces $L$ of a $2l$-dimensional symplectic vector space $G$ over the field $\bb F_2$. More precisely to each Lagrangian there correspond $2^l$ distinct but equivalent stabilizer states. Thus probabilities involving properties of stabilizer states may be expressed using the finite set $\lgrass{(G)}$ of all such Lagrangians as a sample space with Laplace probability measure $P$. For our investigation of asymptotics all dimensions in this set-up will have to be scaled by a large common factor $N$, since each of the $l$ parties will be given not one but a large number $N$ of qubits. The homological invariants of these stabilizer states are then expected to grow of order $N$ too because they behave additively with respect to direct sums \cite{wirthmueller}.

To geometers $\lgrass{(G)}$ is well-known as a Grassmannian, which naturally carries the structure of a projective algebraic variety over the field $\bb F_2$. The classification of stabilizer states by their homological invariants yields a partition of this variety into strata which in turn are quasi-projective varieties. While standard arguments based on the notions of dimension, multiplicity, and more generally the Hilbert polynomial do give asymptotic results concerning the number of points in such varieties they seem to be less of a suitable tool here since the kind of asymptotics is not the right one. The reason is that the context forces us to work over a fixed finite field --- the prime field $\bb F_2$ being the most important one --- and precludes the possibility of passing to finite extensions, let alone the algebraic closure. Nevertheless in the search for generic properties of states we found that the notion of dimension has some predictive value at least.

From yet another point of view the variety $\lgrass{(G)}$ is a homogeneous space of the symplectic group $\Sp(2l,\bb F_2)$, which is a Chevalley group of type $C_l$ over the field $\bb F_2$. This fact makes notions from group theory available which include a method to count the points of $\lgrass{(G)}$ via the Bruhat decomposition into Schubert cells \cite{carter,steinberg} --- indeed this would work over an arbitrary finite field and it determines the Weil zeta function of a Grassmannian. It also is the main tool used in this paper as the relevant strata of $\lgrass{(G)}$ are related to the Schubert cells.

The investigation becomes somewhat simpler if the Lagrange Grassmannian $\lgrass{(G)}$ is replaced by the ordinary one, $\grass{(G,k)}$, which parametrises all $k$-dimensional linear subspaces of the $l$-dimensional vector space $G$ over $\bb F_2$. From the group theoretic point of view we would then be dealing with $\SL(l,\bb F_2)$, a Chevalley group of type $A_{l-1}$ rather than $C_l$. In the context of stabilizer states this situation has its own interest as it arises with CSS states\ts: the CSS state with parameter $L\in\grass{(G,k)}$ corresponds to the Lagrangian
\begin{displaymath}
  L\oplus L^\perp\subset G\oplus G
\end{displaymath}
in $\lgrass{(G\oplus G)}$ where $G\oplus G$ is equipped with the standard symplectic form. We therefore treat this case first before proceeding to the Lagrangian one.

In Section \ref{res} we state qualitative versions of our results as Theorems \ref{css_qual} and \ref{lag_asymptotics}. The first refers to the non-Lagrangian case of $N\mskip-0.7\thinmuskip k$-dimensional subspaces $L\subset\bb F_2^{Nl}$ just mentioned. It roughly states that with arbitrarily high probability, for a large number $N$ of qubits held by each party, nearly all homology of $L$ is concentrated in the single group $H^{l-k+1}L$, and its normalised rank is close to a predetermined number that can be calculated from $l$ and $k$. Similarly Theorem \ref{lag_asymptotics} states that for large $N$, with arbitrarily high probability the homology of a Lagrangian $L\subset\bb F_2^{2Nl}$ is nearly all concentrated in $H^{l+1}L$, and its normalised rank is close to a predetermined number.

The proofs are compiled in Section \ref{pro}, and include much more precise quantitative asymptotic results.
\section{Results}
\label{res}
\noindent
We will work over a fixed finite field $\bb F$ throughout. While only in the case of the prime fields $\bb F_p$ there is a direct correspondence between Lagrangians and stabilizer states (in a space of qu-$p\mskip0.7\thinmuskip$-its), we require no such restriction and allow that $\bb F$ has an arbitrary number $q=p^e$ of elements. We also fix an integer $l\ge2$ while $N\in\bb N$ will be variable\ts; from the point of view of stabilizer states over $\bb F$ we would be dealing with $l$ parties controlling $N$ qudits each. 

The homology $H^j(|\psi\rangle)$ of the CSS state $|\psi\rangle$ is $H^j(L\oplus L^\perp)$ by definition, and since homology is additive by \cite{wirthmueller}\ts Lemma 2 we need only study $H^jL$ for an arbitrary linear subspace $L\subset\bb F^{Nl}$. For the same reason we assume the latter of dimension $N\mskip-0.7\thinmuskip k$ with fixed $k\in\bb N$, so that $\dim H^jL$ will be of order $N$ too. We thus let $\eu G_k(N):=\grass{(\bb F^{Nl},Nk)}$ denote the Grassmannian variety of $N\mskip-0.7\thinmuskip k$-dimensional linear subspaces of the standard vector space $\bb F^{Nl}$, and consider $\eu G_k(N)$ as the sample space of a probability space with Laplacian probability measure $P$.

Recall \cite{wirthmueller} that the construction of the homology groups $H^jL$ starts from a chain complex
\begin{equation}\label{chain_complex}
  \begin{xy}
    \xymatrix@R0.4pc{{C^0L}\ar@{->}[r]^{\delta_0}&
                      {C^1L}\ar@{->}[r]^{\delta_1}&
                      {\cdots}\ar@{->}[r]^(.42){\delta_{j-1}}&
                      {C^jL}\ar@{->}[r]^(.42){\delta_j}&
                      {C^{j+1}L}\ar@{->}[r]^(.54){\delta_{j+1}}&
                      {\cdots}\\}
  \end{xy}
\end{equation}
where the chain group $C^jL$ is spanned by all vectors of $L$ that are local to some subset of $j$ parties. The group $H^jL$ then is the $j$-th cohomology group of that complex, that is $H^jL=\ker\delta_j/\im\delta_{j-1}$ for each $j>0$.

\vspace*{12pt}
\begin{thm}\label{css_qual}
Let $L\in\eu G_k(N)$ be a sample.
\begin{enumerate}
  \item If $j<l\!-\!k$ then $\lim_{N\to\infty}P(C^jL\neq0)=0$.
  \item If $j<l\!-\!k$ or $j>l\!-\!k\!+\!2$ then
    $\lim_{N\to\infty}P(H^jL\neq0)=0$.
\end{enumerate}
Now let any real $\epsilon>0$ be given. Then
\begin{enumerate}
  \setcounter{enumi}{2}
  \item $\lim_{N\to\infty}P({1\over N}\dim C^{l-k}L\ge\epsilon)=0$, and
  \vskip 3pt
  \item $\lim_{N\to\infty}
          P({1\over N}\dim H^{l-k}L+{1\over N}\dim H^{l-k+2}L\ge\epsilon)=0$.
\end{enumerate}
Finally
\begin{enumerate}
  \setcounter{enumi}{4}
  \item $\lim_{N\to\infty}P\left({1\over N}\dim H^jL\ge\epsilon\right)=0$
    if $j\neq l\!-\!k\!+\!1$, and
    \begin{displaymath}
      \lim_{N\to\infty}
        P\left(\bigl|{\textstyle{1\over N}}\dim H^{l-k+1}L-\chi_k\bigr|
         \ge\epsilon\right)
        =0
    \end{displaymath}
    with $\chi_0=0$ and  $\chi_k=\disp{l-2\choose k-1}$ for $k>0$.
\end{enumerate}
\end{thm}
\begin{proof}
Statements 1 and 3 will be immediate corollaries to the more detailed results of Proposition \ref{asymptotics_H}. Furthermore $C^jL=0$ trivially implies $H^jL=0$, and by the duality theorem \cite{wirthmueller}\ts Theorem 5 the group $H^jL$ is the dual of $H^{l-j+2}L^\perp$\ts: thus statement 2 follows from the first, and statement 4 from the third. Finally for each $j$-dimensional coordinate subspace of $\bb F^l$ the intersection of the corresponding coordinate subspace of $\bb F^{Nl}$ with a generic $L$ has dimension $(j\!+\!k\!-\!l)N$ if $j\!+\!k\ge l$, again by Proposition \ref{asymptotics_H}. Therefore
$\sum_{j=l-k+1}^l{(-1)}^j{l\choose j}(j\!+\!k\!-\!l)$
is the expected normalised Euler number of the chain complex (\ref{chain_complex}), and a fortiori that of its cohomology. Using standard identities for binomial coefficients \cite{knuth} the alternating sum is evaluated to ${(-1)}^{l-k+1}\chi_k$.
\end{proof}
\vspace*{12pt}

We now turn to general pure stabilizer states where $L\subset\bb F^{2Nl}$ may be  an arbitrary Lagrangian subspace of $\bb F^{2Nl}$. We let $\eu L(N):=\lgrass{(Nl)}$ be the Grassmannian variety of these subspaces.

\vspace*{12pt}
\begin{thm}\label{lag_asymptotics}
Let $L\in\eu L(N)$ be a sample.
\begin{enumerate}
  \item If $j<l/2$ then $\lim_{N\to\infty}P(C^jL\neq0)=0$.
  \item If $j<l/2$ or $j>l/2\!+\!1$ then
    $\lim_{N\to\infty}P(H^jL\neq0)=0$.
\end{enumerate}
Let any $\epsilon>0$ be given and assume that $l$ is even. Then
\begin{enumerate}
  \setcounter{enumi}{2}
  \item $\lim_{N\to\infty}P({1\over N}\dim C^{l/2}L\ge\epsilon)=0$, and
  \vskip 3pt
  \item $\lim_{N\to\infty}
    P({1\over N}\dim H^{l/2}L+{1\over N}\dim H^{l/2+2}L\ge\epsilon)=0$.
\end{enumerate}
Under the same assumptions let
$\chi=2\chi_{l/2}=2\disp{l-2\choose l/2-1}$
be the absolute value of the expected normalised Euler number of $L$. Then
\begin{enumerate}
  \setcounter{enumi}{4}
  \item $\lim_{N\to\infty}P\left({1\over N}\dim H^jL\ge\epsilon\right)=0$
    if $j\neq l/2\!+\!1$, and
    \begin{displaymath}
      \lim_{N\to\infty}
        P\left(\bigl|{\textstyle{1\over N}}\dim H^{l/2+1}L-\chi\bigr|
        \ge\epsilon\right)=0.
    \end{displaymath}
\end{enumerate}
\end{thm}
\begin{proof}
This follows from Proposition \ref{asymptotics_M} in the same way as Theorem \ref{css_qual} follows from Proposition \ref{asymptotics_H}.
\end{proof}

\noindent
\textit{Remarks.} We do not know whether an analogue of statement 5 holds for odd dimension $l$. ---  As mentioned in the introduction the rank of $H^2(|\psi\rangle)$ is the number of all-party GHZ states that may be extracted from $|\psi\rangle$. By statement 2 of Theorem \ref{lag_asymptotics} this number is zero for $l>4$ and large $N$, with arbitrarily high probability. In fact in the binary case $q=2$ the more detailed results of Proposition \ref{asymptotics_M} imply the bound of $(1\!+\!\log2)\,2^{-(l-4)N}$ for the probability that a GHZ state may be extracted from $|\psi\rangle$. In the case of $l=4$ the number of extractable all-party GHZ states is still likely to be arbitrarily small compared to $N$, by the third statement of Theorem \ref{lag_asymptotics}. Again Proposition \ref{asymptotics_M} explicitly bounds the probability of the event that $N\mskip-0.7\thinmuskip\epsilon$ or more GHZ states may be extracted from $|\psi\rangle$. In the binary case the bound is proportional to $2^{-N^2\epsilon^2}$ and thus improves that given in \cite{smith2} for this particular question, which is exponential in $N$ rather than $N^2$.
\vspace*{12pt}

Theorem \ref{lag_asymptotics} may be read as a statement about the typical entanglement of an $l$-party stabilizer state\ts: all such entanglement comes from the entanglement of the $(l/2\!+\!1)$-party subsets, and it is quantified by the value given in the fifth statement. As to CSS states their typical entanglement as described by Theorem \ref{css_qual} reflects the way their Lagrangian is built from two subspaces $L,L^\perp\subset\bb F^l$ of complementary dimension. Remarkably, in the symmetric case $\dim L=l/2$ their entanglement is the same as that of a general stabilizer state. In physical language this occurs whenever the stabilizer group has the same number of Pauli $X$ and $Z$ generators.
\section{Proofs}
\label{pro}
\noindent
We begin by recalling the well-known partition of the Grassmannian $\grass{(l,k)}$ into Schubert cells. The latter are indexed by right congruence classes of permutations $\sigma\in\Sym_l/(\Sym_k\times\Sym_{l-k})$. More precisely, given any $\sigma\in\Sym_l$ the corresponding Schubert cell $C_\sigma$ comprises those subspaces of $\bb F^l$ which are spanned by the columns of some unipotent upper triangular matrix $u\in\GL(l,\bb F)$, with column index in $\{\sigma1,\sigma2,\dots,\sigma k\}$. Explicitly, if $\sigma$ is chosen in its congruence class such that $\sigma1<\sigma2<\cdots<\sigma k$ then a subspace belongs to $C_\sigma$ if and only if it is the image subspace of some matrix
\begin{displaymath}
  \left\lgroup
    \begin{array}{ccccc}
      \ast              &\ast   &\cdots     &\cdots &\ast   \cr
    \noalign{\vskip-7pt}
      \vdots            &\vdots &           &       &\vdots \cr
    \noalign{\vskip-4pt}
      \ast              &\ast   &\cdots     &\cdots &\ast   \cr
      1                 &0      &\cdots     &\cdots &0      \cr
                        &\ast   &\cdots     &\cdots &\ast   \cr
    \noalign{\vskip-7pt}
                        &\vdots &           &       &\vdots \cr
    \noalign{\vskip-4pt}
                        &\ast   &\cdots     &\cdots &\ast   \cr
                        &1      &0\;\,\cdot &\cdots &0      \cr
                        &       &           &       &\ast   \cr
    \noalign{\vskip-7pt}
                        &       &\ddots     &       &\vdots \cr
    \noalign{\vskip-4pt}
                        &       &           &       &\ast   \cr
                        &       &           &       &1      \cr
      \vbox{\vskip20pt} &       &           &       &       \cr
    \end{array}
  \right\rgroup\in\Mat{(l\times k,\bb F)}
\end{displaymath}
where the special rows are those with indices $\sigma1,\dots,\sigma k$, and the entries at the starred positions are arbitrary. It follows that $C_\sigma$ is an affine space of dimension
\begin{displaymath}
  d_\sigma=
    \left|\left\{(a,b)\in\{1,\dots,k\}\by\{k\!+\!1,\dots,l\}\,\big|\,
            \sigma a>\sigma b\right\}\right|.
\end{displaymath}
From this fact the number of points
\begin{displaymath}
  G_{lk}:=|\grass{(l,k)}|
\end{displaymath}
is easily computed\ts:

\vspace*{12pt}
\begin{prop}\label{formula_G}
  $G_{lk}=\disp{\prod_{i=l-k+1}^l(q^i-1)\over\prod_{i=1}^k(q^i-1)}$\ts.
\end{prop}
\begin{proof}
Clearly $G_{l0}=G_{ll}=1$. For $0<k<l$ the set $\sigma\{1,\dots,k\}$ may or may not contain $l$, and each case corresponds to a summand in the recurrence relation
\begin{displaymath}
  G_{lk}=q^{l-k}G_{l-1,k-1}+G_{l-1,k}.
\end{displaymath}
The formula now follows by induction.
\end{proof}
\vspace*{12pt}

We let $0=\bb F^0\subset\bb F^1\subset\cdots\subset\bb F^{l-1}\subset\bb F^l$ be the standard flag formed by the coordinate subspaces with vanishing \textit{last} components. Then for any $L\in C_\sigma$ we have
\begin{equation}\label{dim_cap}
  \dim L\cap\bb F^j=k-\left|\left\{a\in\{1,\dots,k\}\,\big|\,
                                   \sigma a>j\right\}\right|.
\end{equation}
In particular the collection of these dimensions is an alternative way to characterise the cell $C_\sigma$.

If $j\!+\!k\le l$ then the intersection $L\cap\bb F^j$ generically is zero, and we determine
\begin{equation}\label{def_F}
  F_{lkj}:=\left|\left\{L\in\grass{(l,k)}\,\big|\,L\cap\bb F^j=0\right\}\right|
\end{equation}
in this case.

\vspace*{12pt}
\begin{prop}\label{formula_F}
For $j\!+\!k\le l$ one has
  $F_{lkj}
    =\disp q^{kj}\cdot
       {\prod_{i=l-k-j+1}^{l-j}(q^i-1)\over\prod_{i=1}^k(q^i-1)}$\ts.
\end{prop}
\begin{proof}
From (\ref{dim_cap}) we know that $L\cap\bb F^j=0$ occurs if and only if $L$ belongs to a cell $C_\sigma$ such that $\sigma$ sends $\{1,\dots,k\}$ into $\{j\!+\!1,\dots,l\}$. For $0<k<l$ the set $\sigma\{1,\dots,k\}$ may or may not contain $l$\ts: this gives the recurrence relation
\begin{displaymath}
  F_{lkj}=q^{l-k}F_{l-1,k-1,j}+F_{l-1,k,j},
\end{displaymath}
which together with $F_{l0j}=1$ and $F_{lk0}=G_{lk}$ allows to prove the formula by induction.
\end{proof}
\vspace*{12pt}

For $j\!+\!k\ge l$ the dimension of  $L\cap\bb F^j$ is at least (and generically equal to) $j\!+\!k\!-\!l$, and we extend (\ref{def_F}) to this case putting
\begin{displaymath}
  F_{lkj}:=\left|\left\{L\in\grass{(l,k)}\,\big|\,
                        \dim L\cap\bb F^j=j\!+\!k\!-\!l\right\}\right|
           \q(j\!+\!k\ge l).
\end{displaymath}
The analogue of Proposition \ref{formula_F} is

\vspace*{12pt}
\begin{prop}\label{formula_F_next}
For $j\!+\!k\ge l$ one has
  $F_{lkj}
    =\disp q^{(l-j)(l-k)}\cdot
    {\prod_{i=j+k-l+1}^j(q^i-1)\over\prod_{i=1}^{l-k}(q^i-1)}$\ts.
\end{prop}
\begin{proof}
Passing to orthogonal complements we see that
\begin{displaymath}
  \dim L\cap\bb F^j=j\!+\!k\!-\!l\hbox{\q if and only if\q}
  L^\perp\cap{(\bb F^j)}^\perp=0.
\end{displaymath}
\end{proof}
\vspace*{12pt}

Still in the case $j\!+\!k\ge l$ we now let more generally $s$ be an integer with $j\!+\!k\!-\!l\le s\le\min\{j,k\}$, and calculate the cardinality
\begin{displaymath}
  H_{lkjs}:=\left|\left\{L\in\grass{(l,k)}\,\big|\,
                         \dim L\cap\bb F^j=s\right\}\right|.
\end{displaymath}

\vspace*{12pt}
\begin{prop}\label{formula_H}
  $H_{lkjs}=G_{js}\cdot F_{l-s,k-s,j-s}$
  for $0\le j\!+\!k\!-\!l\le s\le\min\{j,k\}$.
\end{prop}
\begin{proof}
The first factor counts the possible intersections $L':=L\cap\bb F^j$, and the second the possibilities to realise a fixed $L'$\ts: to this end consider the factor spaces $L/L'\in\grass{(L/L',k\!-\!s)}$ with $L/L'\cap\bb F^j/L'=0$.
\end{proof}
\vspace*{12pt}

We are now ready to study the asymptotic behaviour of the entities introduced so far. From the point of view of algebraic geometry the spaces $L\in\grass{(l,k)}$ that have a non-generic, that is larger than expected intersection with some $\bb F^j$ are points of a proper algebraic subvariety of $\grass{(l,k)}$. This implies that the proportion $F_{lkj}:G_{lk}$ approaches $1$ when the ground field $\bb F_q$ is replaced by a large finite extension field. While in our context the ground field is fixed, and rather the dimensions $l$, $k$, and $j$ are scaled by a large common factor $N$ the proportion in question still behaves by and large the same way. The precise result is this\ts:

\vspace*{12pt}
\begin{prop}\label{asymptotics_H}
Let $l$, $k$, and $j$ be fixed.
\begin{enumerate}
  \item If $j\!+\!k\neq l$ then
    \begin{displaymath}
      0\le1-{F_{Nl,Nk,Nj}\over G_{Nl,Nk}}\le
      {|\log(1-q\inv)|\over1-q\inv}\cdot q^{-|l-k-j|N}
    \end{displaymath}
    for all $N\in\bb N$.
    \smallskip
    \item If $j\!+\!k=l$ let any real $\epsilon>0$ be given and put
    \begin{displaymath}
      S_N:=
        \left|\left\{L\in\grass{(Nl,Nk)}\,\big|\,
          \dim L\cap\bb F^{Nj}\ge N\epsilon
        \right\}\right|.
    \end{displaymath}
    Then
    \begin{displaymath}
      {S_N\over G_{Nl,Nk}}\le e^{{2q\over q-1}|\log(1-q\inv)|}
        \cdot\textstyle{{(q^2\!-\!1)}^2\over{(q^2\!-\!1)}^2-q}
        \cdot q^{-N^2\epsilon^2}
    \end{displaymath}
    for all $N\in\bb N$.
\end{enumerate}
\end{prop}
\begin{proof}
We first assume that $j\!+\!k<l$ and put $m=l\!-\!k\!-\!j$. The expressions for $F_{lkj}$ and $G_{lk}$ from Propositions \ref{formula_G} and \ref{formula_F} have the same degree as rational functions in $q$, and the relevant quotient is
\begin{displaymath}
  \begin{array}{rrl}
    Q_N &:= &\disp{F_{Nl,Nk,Nj}\over G_{Nl,Nk}} \\
  \noalign{\vskip 5pt}
        &=  &\disp{\prod_{i=Nl-Nk-Nj+1}^{Nl-Nj}(1-q^{-i})
             \over
             \prod_{i=1}^{Nk}(1-q^{-i})}
             \cdot
             \disp{\prod_{i=1}^{Nk}(1-q^{-i})
             \over
             \prod_{i=Nl-Nk+1}^{Nl}(1-q^{-i})} \\
  \noalign{\vskip 5pt}
        &=  &\disp{\prod_{i=Nm+1}^{Nl-Nj}(1-q^{-i})
             \over
             \prod_{i=Nl-Nk+1}^{Nl}(1-q^{-i})}. \\
  \end{array}
\end{displaymath}
Giving away the denominator we abbreviate $u=-q\!\cdot\!\log(1\!-\!q\inv)$, and using $m>0$ further estimate
\begin{equation}
  \begin{array}{rcl}\label{geom_series}
    \log Q_N &\ge &\disp\sum_{i=Nm+1}^{Nl-Nj}\log(1-q^{-i}) \\
  \noalign{\vskip 3pt}
             &\ge &-u\cdot\!\!\disp\sum_{i=Nm+1}^{Nl-Nj}q^{-i} \\
  \noalign{\vskip 3pt}
             &\ge &-u\cdot\!\disp\sum_{i=Nm+1}^\infty q^{-i} \\
  \noalign{\vskip 3pt}
             &= &-u\cdot q^{-Nm-1}\disp{1\over 1-q\inv} \\
  \end{array}
\end{equation}
for all $N$. We conclude
\begin{displaymath}
  1-Q_N\le1-\exp\bigl(-{u\over q-1}\,q^{-mN}\bigr)
       \le {u\over q-1}\,q^{-mN}
\end{displaymath}
and need only substitute the value of $u$.

The case of $j\!+\!k>l$ is similar, and there remains that of $j\!+\!k=l$. Here Proposition \ref{formula_H} supplies the quotient
\begin{displaymath}
  \disp{H_{Nl,Nk,Nj,\sigma}\over H_{Nl,Nk,Nj,\sigma-1}}
    =q^{2\sigma-Nl-1}\cdot
       \disp{(q^{Nk-\sigma+1}-1)(q^{Nj-\sigma+1}-1)\over{(q^\sigma-1)}^2}
\end{displaymath}
for each $\sigma>0$. In view of $H_{Nl,Nk,Nj,0}=F_{Nl,Nk,Nj}$ we obtain the estimate
\begin{displaymath}
  \disp{H_{Nl,Nk,Nj,s}\over F_{Nl,Nk,Nj}}
    \le\disp{q^s\over{(q\!-\!1)}^2{(q^2\!-\!1)}^2\cdots{(q^s\!-\!1)}^2}
\end{displaymath}
and, summing up,
\begin{displaymath}
  \begin{array}{rcl}
    \disp\sum_{\sigma=s}^{\min{\{Nj,Nk\}}}
           {H_{Nl,Nk,Nj,\sigma}\over F_{Nl,Nk,Nj}}
      &\le &\disp{q^s\over{(q\!-\!1)}^2{(q^2\!-\!1)}^2\cdots{(q^s\!-\!1)}^2}
              \cdot\sum_{i=0}^\infty{q^i\over{(q^{s+1}-1)}^{2i}} \\
      &=   &\disp{q^s\over{(q\!-\!1)}^2{(q^2\!-\!1)}^2\cdots{(q^s\!-\!1)}^2}
              \cdot{{(q^{s+1}\!-\!1)}^2\over{(q^{s+1}\!-\!1)}^2-q}   \\
  \end{array}
\end{displaymath}
for all $s>0$. In the resulting estimate
\begin{displaymath}
  \begin{array}{rcl}
    \disp\sum_{\sigma\ge s}\disp{H_{Nl,Nk,Nj,\sigma}\over G_{Nl,Nk}}
      &\le &\disp{q^s\over{(q\!-\!1)}^2{(q^2\!-\!1)}^2\cdots{(q^s\!-\!1)}^2}
              \cdot{{(q^2\!-\!1)}^2\over{(q^2\!-\!1)}^2-q} \\
  \noalign{\vskip 3pt}
      &=   &\disp{1\over{(1\!-\!q^{-1})}^2{(1\!-\!q^{-2})}^2
                      \cdots{(1\!-\!q^{-s})}^2}
                      \cdot{{(q^2\!-\!1)}^2\over{(q^2\!-\!1)}^2-q}
                      \cdot q^{-s^2}                       \\
  \end{array}
\end{displaymath}
the product $(1-q^{-1})\cdot\,\cdots\,\cdot(1-q^{-s})$ is bounded below by $e^{-u/(q-1)}$ as in (\ref{geom_series}), so that
\begin{displaymath}
    \disp\sum_{\sigma\ge s}\disp{H_{Nl,Nk,Nj,\sigma}\over G_{Nl,Nk}}
      \le e^{2u/(q-1)}\cdot{{(q^2\!-\!1)}^2\over{(q^2\!-\!1)}^2-q}
                      \cdot q^{-s^2}.
\end{displaymath}
It only remains to substitute $s:=\lceil N\epsilon\rceil$.
\end{proof}

\noindent
\textit{Remark.} Statement 1 of the proposition definitely does not extend to the case $j\!+\!k=l$, for the quotient
\begin{displaymath}
  {F_{Nl,Nk,Nj}\over G_{Nl,Nk}}=\prod_{i=1}^{Nk}{1-q^{-i}\over1-q^{-Nj-i}}
\end{displaymath}
clearly cannot converge to one for $N\to\infty$ if $j$ is positive.
\vspace*{12pt}

We now turn to the symplectic vector space $\bb F^{2l}$ and let $\lgrass{(2l)}$ denote the Grassmannian of its Lagrangian subspaces. As mentioned in the introduction $\lgrass{(2l)}$ is a homogeneous space of the symplectic group $\Sp(2l,\bb F)$, which is a Chevalley group of type $C_l$.

Writing $e_a$ for the $a$-th standard column vector in $\bb R^{2l}$ the vectors
\begin{displaymath}
  e_a-e_b\hbox{ for }a<b\hbox{, \ and }e_a+e_b\hbox{ for }a\le b
\end{displaymath}
form a system of positive roots for $\Sp(2l,\bb F)$. The corresponding Borel subgroup $B$ consists of all matrices
\begin{equation}\label{borel_matrix}
  \left\lgroup\matrix{
    g & gs               \cr
    0 & {(g\transp)}\inv \cr
  }\right\rgroup
   \in\Mat{(2l\times2l,\bb F)}\hbox{ with }g\in B'\hbox{ and }
     s\in\Sym{(l,\bb F)}
\end{equation}
where $B'\subset\GL(l,\bb F)$ is the ordinary Borel group of upper triangular matrices, and $\Sym{(l,\bb F)}$ the space of symmetric matrices of that size. The elements of the Weyl group $W={\{\pm1\}}^l\Sym_l$ act as permutations $\sigma$ of the set $\{\pm1,\dots,\pm l\}$ with the property that $\sigma(-a)=-\sigma a$ for all $a$.

The  Schubert cells of $\lgrass{(2l)}$ are indexed by the right congruence classes in $W/\Sym_l$\ts; they are in one-to-one correspondence with the normalised permutations $\sigma$ --- those which satisfy the condition
\begin{displaymath}
  \sigma1<\cdots<\sigma l
\end{displaymath}
and thus are shortest in their congruence classes (recall that $\sigma$ acts on the set $\{\pm1,\dots,\pm l\}$ of \textit{signed} integers). The cells are simply characterised by the set $\left\{a\in\sigma\{1,\dots,l\}\,\big|\,a<0\right\}$, and the length
\begin{equation}\label{lag_length}
  d_\sigma=
    \left|\left\{(a,b)\in\{1,\dots,l\}^2\,\big|\,
            a\le b\hbox{ and }\sigma a+\sigma b<0\right\}\right|
\end{equation}
is the dimension of the Schubert cell $C_\sigma$ labelled by the class of $\sigma$. Of course $C_\sigma$ may be described explicitly, as follows. Define $n$ by $\sigma n<0<\sigma(n\!+\!1)$ and put $p=l\!-\!n$. Then a subspace of $\bb F^{2l}$ belongs to $C_\sigma$ if and only if it is the image subspace of a (unique) matrix
\begin{equation}\label{sp_matrix}
  \left\lgroup
    \begin{array}{ccccl | rcccc}
      \ast              &\ast &\cdots &\cdots &\ast
         &\circ        &\cdots &\cdots &\cdots &\circ   \cr
      1                 &0    &\cdots &\cdots &0
         &0            &\cdots &\cdots &\cdots &0       \cr
                        &\ast &\ast   &\cdots &\ast
         &\circ        &\cdots &\cdots &\cdots &\circ   \cr
                        &1    &0      &\cdots &0
         &0            &\cdots &\cdots &\cdots &0       \cr
                        &     &       &       &\ast
         &\circ        &\cdots &\cdots &\cdots &\circ   \cr
    \noalign{\vskip-7pt}
                        &     &\ddots &       &\,\vdots\,\,
         &\,\,\vdots\, &       &       &\      & \vdots \cr
    \noalign{\vskip-4pt}
                        &     &       &       &\ast
         &\circ        &\cdots &\cdots &\cdots &\circ   \cr
                        &     &       &       &1
         &0            &\cdots &\cdots &\cdots &0       \cr
      \strut           &     &       &       &\strut
          &\circ       &\cdots &\cdots &\cdots &\circ  \cr
      \hline
      &&&& \strut           &              &       &       &     &     \cr
      &&&&                   &1            &       &       &     &     \cr
      &&&&                   &\ast         &       &       &     &     \cr
    \noalign{\vskip-7pt}
      &&&&                   &\,\,\vdots\, &       &\ddots &     &     \cr
    \noalign{\vskip-4pt}
      &&&&                   &\ast         &       &       &     &     \cr
      &&&&                   &0            &\cdots &\cdot\;\,0 &1    &     \cr
      &&&&                   &\ast         &\cdots &\cdots &\ast &     \cr
      &&&&                   &0            &\cdots &\cdots &0    &1    \cr
      &&&&                   &\ast         &\cdots &\cdots &\ast &\ast \cr
    \end{array}
  \right\rgroup\in\Mat{(2l\times(p\!+\!n),\bb F)}
\end{equation}
where steps occur at the row indices
\begin{displaymath}
  l\!+\!1\!-\!\sigma b\hbox{ with }\sigma b>0\hbox{\q and\q}
  2l\!+\!1\!-\!|\sigma a|\hbox{ with }\sigma a<0,
\end{displaymath}
and where the unspecified entries are subject to the following conditions. All starred entries in either the upper left hand or the lower right hand submatrix may be chosen arbitrarily, and the other half are then determined by the condition that these submatrices span mutually orthogonal subspaces of $\bb F^l$. The circled entries are free within a symmetry condition dependent on the starred ones as in (\ref{borel_matrix}).

It is well-known how to count Lagrangians\ts:

\vspace*{12pt}
\begin{prop}\label{formula_L} There are exactly $L_l=\prod_{i=1}^l(q^i+1)$ Lagrangian subspaces of $\bb F^{2l}$, and more generally the number of $k$-dimensional isotropic subspaces is
\begin{displaymath}
  L_{lk}={\prod_{i=l-k+1}^l(q^{2i}-1)\over\prod_{i=1}^k(q^i-1)}
\end{displaymath}
for $k\le l$.
\end{prop}
\begin{proof}
For $l>0$ the set $\sigma\{1,\dots,l\}$ contains either $-l$ or $l$, so (\ref{lag_length}) gives the recursive formula $L_l=q^lL_{l-1}+L_{l-1}$ for $L_l$.

For general $k\le l$ we make two observations\ts: given an isotropic subspace $K\subset\bb F^{2l}$ of dimension $k$ the Lagrangians of $\bb F^{2l}$ that contain $K$ are in one-to-one correspondence with Lagrangian subspaces of $K^\perp/K$, while on the other hand given a Lagrangian $L\subset\bb F^{2l}$ every subspace of $L$ is isotropic. Thus counting pairs $(L,K)$ we obtain the identity
\begin{displaymath}
  L_{l-k}\cdot L_{lk}=L_l\cdot G_{lk},
\end{displaymath}
and the result follows using Proposition \ref{formula_G}.
\end{proof}
\vspace*{12pt}

For a given integer $j$ with $0\le j\le l$ we let
\begin{displaymath}
  E^j\oplus F^j\subset\bb F^l\oplus\bb F^l
\end{displaymath}
denote the sum of the subspaces spanned by the first $j$ coordinates in $\bb F^l\oplus0$, respectively in $0\oplus\bb F^l$. In order to compute
\begin{displaymath}
  K_{lj}:=\left|\left\{L\in\lgrass{(2l)}\,\big|\,L\cap E^j=0\right\}\right|
\end{displaymath}
we read off from (\ref{sp_matrix}) that
\begin{displaymath}
  \dim L\cap E^j=\left|\left\{b\in\{1,\dots,l\}\,\big|\,
                   \sigma b>l\!-\!j\right\}\right|,
\end{displaymath}
so that $K_{lj}$ counts the cases with $\sigma b\le l\!-\!j$ for all $b$, or equivalently
\begin{equation}\label{e_condition}
  \sigma a=a\!-\!1\!-\!l\hbox{\q for }a=1,\dots,j.
\end{equation}
In particular assuming $l$ and $j$ positive, the condition can only be met if $-l$ rather than $l$ belongs to $\sigma\{1,\dots,l\}$, so we have $K_{lj}=q^lK_{l-1,j-1}$ and conclude\ts:

\vspace*{12pt}
\begin{prop}\label{formula_K}
  $K_{lj}=q^{j(2l-j+1)/2}\cdot\disp\prod_{i=1}^{l-j}(q^i+1)$.
\end{prop}
\vspace*{12pt}

We now compute
$J_{lj}:=\left|\left\{L\in\lgrass{(2l)}\,\big|\,L\cap(E^j+F^j)=0\right\}\right|$.

\vspace*{12pt}
\begin{prop}\label{formula_J_K}
Assume $2j\le l$ and let $\sigma$ be a normalised permutation with $\sigma b\le l\!-\!j$ for all $b$. The proportion
\begin{displaymath}
  \left|\left\{L\in C_\sigma\,\big|\,L\cap(E^j+F^j)=0\right\}\right|
   :\left|\left\{L\in C_\sigma\,\big|\,L\cap E^j=0\right\}\right|
\end{displaymath}
is the same for all such $\sigma$ and equal to
$\prod_{i=l-2j+1}^{l-j}(1-q^{-i})$.
In particular
\begin{displaymath}
  {\textstyle J_{lj}\over\textstyle K_{lj}}
    =\disp\prod_{i=l-2j+1}^{l-j}(1-q^{-i}).
\end{displaymath}
\end{prop}
\begin{proof}
Fix $\sigma$ as in the hypothesis and let $L\in C_\sigma$ be the Lagrangian with spanning matrix (\ref{sp_matrix}). In view of (\ref{e_condition}) only the columns with index $p\!+\!1,\dots,p\!+\!j$ can contribute to the intersection $L\cap(E^j+F^j)$, and the submatrix comprising these columns has the form
\begin{equation}\label{right_sp_matrix}
  \left\lgroup
    \begin{array}{c}
      e' \cr
      e'' \cr
      \hline
      1 \cr
      f'' \cr
    \end{array}
  \right\rgroup\in\Mat{(2l\times j,\bb F)}
\end{equation}
with blocks $e'$ and the unit matrix $1$ of height $j$, and $e'',f''$ of height $l\!-\!j$. Furthermore, $e''$ has $p$ zero rows (with indices $l\!+\!1\!-\!\sigma b$ where $n<b\le l$) and $f''$ has $n\!-\!j$ zero rows (with indices $2l\!+\!1\!-\!|\sigma a|$ where $j<a\le n$, all indices counted with respect to the full matrix, on a scale from 1 to $2l$). The remaining rows of $e''$ and $f''$ are free and together form an $(l\!-\!j)\by j$ matrix $g$ which determines whether $L\cap(E^j+F^j)=0$\ts: this occurs if and only if $\rk g=j$.

Thus the proportion in question is that of the matrices of full rank amongst all $(l\!-\!j)\by j$ matrices, and therefore equal to
\begin{displaymath}
  \begin{array}{rcl}
    {\textstyle|\GL(l\!-\!j,\bb F)|\over
     \textstyle q^{j(l-2j)}\cdot|\GL(l\!-\!2j,\bb F)|}:q^{j(l-j)}
      &= &q^{-j(2l-3j)}\cdot
           {\textstyle q^{{(l-j)}^2}\;\cdot\prod_{i=1}^{l-j}\,(1-q^{-i})\over
            \textstyle q^{{(l-2j)}^2}\cdot\prod_{i=1}^{l-2i}(1-q^{-i})} \\
  \noalign{\vskip 3pt}
      &= &\prod_{i=l-2j+1}^{l-j}(1-q^{-i}).                             \\
  \end{array}
\end{displaymath}
\end{proof}
\vspace*{12pt}

We extend the definition of $J_{lj}$ to the case of $2j\ge l$ in terms of the more general
\begin{displaymath}
  M_{ljs}:=\left|\left\{L\in\lgrass{(2l)}\,\big|\,
    \dim L\cap(E^j+F^j)=s\right\}\right|
\end{displaymath}
by putting $J_{lj}=M_{lj,2j-l}$ if $2j\ge l$.

\vspace*{12pt}
\begin{cor}\label{formula_J}
We have
\begin{displaymath}
  J_{lj}
    =\disp q^{j^2}\cdot
    {\prod_{i=1}^{l-j}(q^{2i}-1)\over\prod_{i=1}^{l-2j}(q^i-1)}
    \hbox{\q if }2j\le l,
\end{displaymath}
and more generally
\begin{displaymath}\label{formula_M}
  M_{ljs}
    =\disp q^{{(j-s)}^2}\cdot
     {\textstyle\prod_{i=j-s+1}^j(q^{2i}\!-\!1)\cdot
                \prod_{i=1}^{l-j}(q^{2i}\!-\!1)\over
      \textstyle\prod_{i=1}^s(q^i\!-\!1)    \cdot
                \prod_{i=1}^{l+s-2j}(q^i\!-\!1)}
    \hbox{\q for }\max\{0,2j\!-\!l\}\le s\le j.
\end{displaymath}
\end{cor}
\begin{proof}
The first formula results from Propositions \ref{formula_K} and \ref{formula_J_K}, and it implies the second by the following reasoning.

For the Lagrangians $L$ to be counted the intersection $L':=L\cap(E^j+F^j)$ is an $s$-dimensional isotropic subspace of $E^j\oplus F^j$, and Proposition \ref{formula_L} gives the number $L_{js}$ of possible $L'$. Now fixing $L'$ we determine the number of Lagrangians $L$ with $L\cap(E^j+F^j)=L'$, and for this purpose may assume $L'=E^s$. Such $L$ intersect $F^s$ trivially, and therefore, under the projection $E^l\oplus F^l\to E^l/E^s\oplus F^l/F^s$ they correspond to Lagrangians $\overline L\subset E^l/E^s\oplus F^l/F^s$ with $\overline L\cap(E^j/E^s\oplus F^j/F^s)=0$. Thus there are exactly $J_{l-s,j-s}$ such $L$ for any given $L'$, and we obtain the formula $M_{ljs}=L_{js}\cdot J_{l-s,j-s}$ and thereby the result.
\end{proof}
\vspace*{12pt}

We finally state the asymptotic properties of the data we have collected.

\vspace*{12pt}
\begin{prop}\label{asymptotics_M}
Let $l$ and $j$ be fixed.
\begin{enumerate}
  \item If $2j\neq l$ then
    \begin{displaymath}
      0\le1-{J_{Nl,Nj}\over L_{Nl}}\le{q\inv+|\log(1-q\inv)|\over1-q\inv}
        \cdot q^{-|l-2j|N}
    \end{displaymath}
    for all $N\in\bb N$.
    \smallskip
    \item Assume that $2j\le l$ and, for any given $\epsilon>0$ put
       \begin{displaymath}
      T_N:=
        \left|\left\{L\in\lgrass{(2Nl)}\,\big|\,
          \dim L\cap(E^{Nj}\oplus F^{Nj})\ge N\epsilon
        \right\}\right|.
    \end{displaymath}
    Then
    \begin{displaymath}
      {T_N\over L_{Nl}}\le
        e^{{2q\over q-1}|\log(1-q\inv)|}
          \cdot\textstyle{{(q^2\!-\!1)}^2\over{(q^2\!-\!1)}^2-q}\cdot q^{-N^2\epsilon^2}
    \end{displaymath}
    for all $N\in\bb N$.
\end{enumerate}
\end{prop}
\begin{proof}
Quite analogous to that of Proposition \ref{asymptotics_H}.
\end{proof}
\vspace*{12pt}

%
%
%

\end{document}